\documentclass{article}

\usepackage{graphicx} 
\usepackage[margin=1in]{geometry}
\usepackage{algpseudocode}
\usepackage{algorithm, comment}
\usepackage{amssymb}
\usepackage{amsthm}
\usepackage{amsmath}
\newtheorem{theorem}{Theorem}
\newtheorem{lemma}{Lemma}

\title{Efficient Matrix Factorization Via Householder Reflections}
\author{
  \begin{tabular}{cc}
    Anirudh Dash & Aditya Siripuram \\
    \small Department of Electrical Engineering & \small Department of Electrical Engineering \\
    \small Indian Institute of Technology, Hyderabad & \small Indian Institute of Technology, Hyderabad \\
    \small Hyderabad, Telangana, India & \small Hyderabad, Telangana, India \\
    \small \texttt{ee21btech11002@iith.ac.in} & \small \texttt{staditya@ee.iith.ac.in}
  \end{tabular}
}
\date{}

\newcommand{\keywords}[1]{\textbf{Keywords:} #1}

\begin{document}

\maketitle

\begin{abstract}
    Motivated by orthogonal dictionary learning problems, we propose a novel method for matrix factorization, where the data matrix $\mathbf{Y}$ is a product of a Householder matrix $\mathbf{H}$ and a binary matrix $\mathbf{X}$. First, we show that the exact recovery of the factors $\mathbf{H}$ and $X$ from $\mathbf{Y}$ is guaranteed with $\Omega(1)$ columns in $\mathbf{Y}$. Next, we show approximate recovery (in the $l_\infty$ sense) can be done in polynomial time($O(np)$) with  $\Omega(\log n)$ columns in $\mathbf{Y}$. We hope the techniques in this work help in developing alternate algorithms for orthogonal dictionary learning.
\end{abstract}


\section{Introduction}

The orthogonal dictionary factorization problem is posed as follows: Given a matrix $\mathbf{Y} \in \mathbb{R}^{n \times p}$, can we find an orthogonal matrix $\mathbf{V} \in \mathbb{R}^{n \times n}$ and a coefficient matrix $\mathbf{X} \in \mathbb{R}^{n \times p}$ such that $\mathbf{Y=VX}$. Variants of this problem appear in standard sparse signal processing literature \cite{elad2010sparse} and signal processing-based graph learning approaches \cite{dong2019learning}, \cite{thanou2014learning}. In the latter, the orthogonal matrix $\mathbf{V}$ is known to be an eigenvector matrix of a suitable graph; in particular, we note that the orthogonal matrix has some additional structure. The goal of this work is to investigate recovery guarantees on $\mathbf{V}$ and $\mathbf{X}$ under strong structural assumptions on the orthogonal matrix.

The standard unstructured dictionary learning problem ($\mathbf{Y=DX}$) has been well investigated in literature. The seminal work of Olshausen et. al in 1997 \cite{olshausen1997sparse} involved recovering sparse representations of an image. Since then, extensive research has led to several fascinating methods to solve the aforementioned problem. Engan et al. \cite{engan1999method} proposed the method of optimal directions (MOD), while Aharon et al. \cite{aharon2006k} proposed the K-SVD method for sparse representations of signals. Mairal et. al provided an online dictionary learning method for sparse coding \cite{mairal2009online}, which was further improved by adding sparsity constraints. It has been shown \cite{sun2015complete} that recovering $\mathbf{D}$ and $\mathbf{X}$ under the assumption of a Bernoulli-Gaussian model on $\mathbf{X}$ for an arbitrary dictionary can be reduced to the orthogonal dictionary learning problem ($\mathbf{Y=VX}$).

New algorithms for orthogonal dictionary learning, based on alternate minimization, were proposed by Arora et. al \cite{arora2014new}, \cite{agarwal2014learning}, while some results obtained on studying local identifiability \cite{geng2014local}, \cite{schnass2015local} have also been in the mix. Li et. al \cite{li2017provable} used alternating gradient descent for non-negative matrix factorization with strong correlations. Traditionally, $l_0$ minimization and its convex relaxation- $l_1$ minimization have been widely used in this field (Spielman et al. \cite{spielman2012exact}). Recently, Zhai et al. \cite{zhai2020complete} developed an iterative approach for complete dictionary learning via $l_4$ minimization over an orthogonal group.

Theoretical results pertinent to the above problem are usually of two kinds: proving the validity of proposed algorithms, and identifying fundamental conditions (i.e. number of columns $p$ required) for any algorithm to recover the factors $\mathbf{V}$ and $\mathbf{X}$.

This work focuses on the problem of orthogonal dictionary learning and is motivated by the following observations:
\begin{enumerate}
    \item Some applications, specifically graph learning, place additional structural assumptions on the orthogonal dictionary
    \item Most of the existing work is on unstructured orthogonal dictionary learning
    \item Most of the existing techniques are iterative, and are sensitive to initialization \cite{liang2022simple}
\end{enumerate}

This naturally motivates us to investigate the effect of significant structural constraints on the orthogonal matrix. We start the above investigation by assuming that the orthogonal matrix is a Householder matrix. We note that every orthogonal matrix can be expressed as a product of Householder matrices \cite{golub2013matrix, uhlig2001constructive}, thus allowing for the development of a new procedure to solve the orthogonal dictionary factorization problem in the future. Solutions to this problem could potentially lead to a new set of non-iterative and initialization-free algorithms to solve the above problem.

In this paper, we analyze the fundamental unit of orthogonal matrices, Householder matrices, as a first step to solving the general orthogonal dictionary factorization problem. We utilize the correlation between the entries of a Householder matrix to develop an algorithm which can recover the factor matrices completely deterministically. As compared to the $\Omega(n \log n)$ \cite{spielman2012exact} bound on the number of columns in the coefficient matrix for orthogonal dictionary learning, we show that recovery of Householder matrices is possible in $\Omega(1)$ columns; albeit with exponential time. To recover the matrices in polynomial time, we allow for some small errors. We show that (with a Bernoulli model on the coefficient matrix), not only is the recovery possible in polynomial time, but it achieves the $\Omega(\log n)$ columns bound in the $l_{\infty}$ sense. It is also an 'all at a time' recovery, i.e., it doesn't find the matrices one column at a time. This paves the way for a new approach to possibly solve the orthogonal dictionary factorization problem. As compared to previous methods, our technique is completely free of initialization and is a non-iterative approach with theoretical guarantees for recovery. The computational complexity involved in the calculation of the factors is $O(np)$, which is substantially smaller than previous methods such as \cite{zhai2020complete}.

\section{Problem Formulation}
We describe the 2 primary scenarios under which we solve our problem.
\subsection{Matrix Recovery for General Binary Matrix}

Consider the following setup:
Following up on the introduction, consider the product
$\mathbf{Y}=\mathbf{H}\mathbf{X} $, where $\mathbf{H}=\mathbf{I}-2\mathbf{uu^T}$ is an orthogonal matrix for some (unknown) unit vector $\mathbf{u}$. We refer by $u_i$ the entries of $\mathbf{u}$. Given the data matrix $\mathbf{Y}$, we want to try and estimates $\hat{\mathbf{H}}$ (and consequently the vector $\hat{u}$) and $\hat{\mathbf{X}}$. We also want to establish why this is tractable for a binary matrix $\mathbf{X}$ but not for a general matrix. Our error metric is \(e=\lVert \mathbf{u}-\hat{\mathbf{u}} \rVert_{\infty} = \max |u_i - \hat{u}_i|\). Some of our results pertain to perfect recovery, i.e. $\mathbf{u} = \hat{\mathbf{u}}$. \\
Note that we say $f(n)= \Omega_r(g(n))$ if $f(n) \geq C_r g(n)$ for some constant $C_r$ depending  only on $r$; for all $n$ large enough.

\subsection{Matrix Recovery for Bernoulli Matrix}
We make the following changes to the previous case:
Define a sparsity model on $\mathbf{X}$, where elements are filled by drawing from an iid Bernoulli distribution with a parameter $\theta$.
\begin{equation}
    \begin{aligned}
        X_{ij} = 1 \text{ w.p. } \theta, 0 \text{ w.p. } 1-\theta.
    \end{aligned}
    \label{eq:1*}
\end{equation}

\section{Main Results}
\subsection{Matrix Recovery for General Binary Matrix}
\begin{theorem}  (\emph{Zero error achievability}) For the general model, $\mathbf{Y}=\mathbf{H}\mathbf{X}$, where $\mathbf{H}=\mathbf{I}-2\mathbf{uu^T}$ and $\mathbf{X}$ is an arbitrary binary matrix, $\mathbf{X}$ can be \emph{uniquely} recovered with $p=\Omega(1)$ columns in $\mathbf{Y}$ (In fact, just two (distinct) columns suffice). Note that there is no assumption of a probabilistic model for the entries of $\mathbf{X}$.
    \label{thm:1}
\end{theorem}
\begin{theorem} If $\mathbf{X}$ is an arbitrary (non-binary) matrix, $\mathbf{H}, \mathbf{X}$ \emph{cannot} be \emph{uniquely} recovered (even up to permutation and sign) with any
    number of columns $p$. Thus, the assumption that $\mathbf{X}$ is a binary matrix is justified for recovery of the Householder dictionary.  \end{theorem}

\subsection{Matrix Recovery for Bernoulli Matrix}

\begin{lemma}
    \label{lem:1}
    (\emph{Parameter Recovery}): Consider the parametric model on the coefficient matrix described in \eqref{eq:1*}. For \emph{any} orthogonal matrix $\mathbf{V}$, the Bernoulli parameter $\theta$ can be recovered with the probability of error upper bounded by the following relation:
    \begin{align*}
        \mathbb{P}\left[ \lvert \hat{\theta} - \theta \rvert > t \right]
         & \leq 2 \exp \left(-2t^2np \right).
    \end{align*}
\end{lemma}
Thus, recovery with high accuracy is possible in $\Omega(1)$ columns. The computational complexity of calculating $\hat{\theta}$ is $O(np)$.

\begin{theorem}
    \label{thm:2}
    (\emph{Householder Recovery}) Consider the parametric model described in \eqref{eq:1*}. Then, the unit vector $\mathbf{u}$ defining the Householder matrix $\mathbf{H}$ can be recovered with high accuracy with the following recovery guarantee
    \[
        \mathbb{P} \left ( \lVert \mathbf{u} - \hat{\mathbf{u}} \rVert_\infty > t\right ) \leq 2 n\exp \left(-8t^2 c^2 \theta^2 p \right)
    \]
    where \(c=\sum_{i=1}^{n} {u_i}\) and \(c \neq 0\). Thus for
    \[p=\Omega\left( \frac {\log (2n^2)}{8t^2\theta^2c^2} \right), \quad  \mathbb{P} \left ( \lVert \mathbf{u} - \hat{\mathbf{u}} \rVert_\infty > t\right ) \rightarrow 0. \]
    The computational complexity involved in calculating $\mathbf{\hat{u}}$ is $O(np)$.
\end{theorem}

\section{Algorithms}
The first algorithm pertains to zero error recovery of the Householder matrix $\mathbf{H}$ and the corresponding unit vector $\mathbf{u}$.We show that a brute-force elimination of possibilities on the columns of $\mathbf{X}$ (see Algorithm \ref{find_HX_exponential}) uniquely identifies the vector $\mathbf{u}$ (Theorem \ref{thm:1}). We show that any solution to the problem $\mathbf{Y}=\mathbf{H}\mathbf{X}$ is unique if it exists.
\begin{algorithm}[!ht]
    \caption{Finding $\mathbf{H}, \mathbf{X}$ for $\mathbf{Y=HX}$ with zero error}
    \label{find_HX_exponential}
    \vspace{2pt}
    \textbf{Input:} $\mathbf{Y}$ \\
    \textbf{Output:} $\mathbf{H}, \mathbf{X}$
    \vspace{10pt}
    \begin{algorithmic}[1]
        \While{all $n$ length binary vectors have not been exhausted}
        \State Set the first column of $\mathbf{X}$ as a random n length binary vector
        \State Find $\mathbf{u}$
        \EndWhile
        \vspace{4pt}
        \State Repeat the above for the second column of $\mathbf{X}$
        \vspace{4pt}
        \State If any of the $\mathbf{u}$ vectors obtained from a vector in the first column
        of $\mathbf{X}$ match with that of a $\mathbf{u}$ vector obtained from a random vector in the second column
        of $\mathbf{X}$, then set $\mathbf{H}$ as the corresponding Householder matrix, which is generated from $\mathbf{u}$
    \end{algorithmic}
\end{algorithm}
The rest of our algorithms operate by exploiting the correlation in the entries of $\mathbf{Y}= \mathbf{H}\mathbf{X}$ that manifest due to the Householder structure of the matrix $\mathbf{H}$ (see Algorithm \ref{find_V_X}). The proofs involve using suitable concentrations on the empirical correlations.
\begin{algorithm}[h]
    \caption{Finding $\theta$ for $\mathbf{Y=VX}$}
    \label{find_theta}
    \vspace{2pt}
    \textbf{Input:} $\mathbf{Y}$ \\
    \textbf{Output:} $\theta$
    \vspace{10pt}
    \begin{algorithmic}[1]
        \State Set $\theta = \frac{\sum_{i=1}^{n}\sum_{j=1}^{p} \mathbf{Y}_{ij}^2}{(n)(p)}$
        \vspace{10pt}
        \State This is the estimate for $\theta$
    \end{algorithmic}
\end{algorithm}

\begin{algorithm}[!ht]
    \caption{Finding $\mathbf{H, \ X}$ for $\mathbf{Y=HX}$ in polynomial time}
    \label{find_V_X}
    \vspace{2pt}
    \textbf{Input:} $\mathbf{Y}$ \\
    \textbf{Output:} $\mathbf{H, \ X}$
    \vspace{10pt}
    \begin{algorithmic}[1]
        \State Algorithm 1 gives $\theta$
        \State Set $(c^2) = \left(\frac{-1}{2}\right)\left(\frac{\sum_{i=1}^{n}\sum_{j=1}^{p} \mathbf{Y_{ij}}}{p\theta}-n\right) $
        \While{$i$ in rows of $\mathbf{Y}$}
        \vspace{4pt}
        \State set $k_i= \frac{-1}{2}\left(\frac{\sum_{j=1}^{p}
                \mathbf{(Y)_{ij}}}{\theta(1-2u_ic)}-1\right)$
        \EndWhile
        \While{$i$ in rows of $\mathbf{Y}$}
        \vspace{4pt}
        \State set $(u)_i= \frac{k_i}{ \sqrt{\sum_{j=1}^{n}
                    k_j}}$
        \EndWhile
        \State Set $\mathbf{H} = \mathbf{I-2uu^T}$
        \State Set $\mathbf{X'} = \mathbf{H}^T\mathbf{Y}$
        \State Set $\mathbf{X} =$ $HT_{\zeta}$($\mathbf{X'}$)
    \end{algorithmic}
\end{algorithm}

\noindent \textbf{Remark} In Algorithm \ref{find_V_X}, $HT_{\zeta}$ is the hard threshold operator, with $\zeta=0.5$. This value has been chosen heuristically. Furthermore, we are implicitly using the following: if $\mathbf{u}$ is a solution, then $-\mathbf{u}$ is also a solution, as both produce the same householder matrix.

\section{Simulations} 
Figure \ref{fig:inf_error} illustrates the $l_{\infty}$ error in $\mathbf{u}$ obtained from experimental results. The ground truth $\mathbf{H}$ and $\mathbf{X}$ matrices are generated arbitrarily as per the conditions specified above. The number of rows is set to $n=1000$ and the number of columns is varied from $2$ onwards. The $y$-axis denotes the $l_{\infty}$ error in $\mathbf{u}$, and the $x$-axis represents the number of columns in $\mathbf{Y}$. The plots are generated for two values of $\theta$, $0.1$ and $0.4$. The error decreases with an increase in the number of columns, as expected theoretically. The error when $\theta=0.1$ is slightly higher than the corresponding error for $\theta=0.4$, which is consistent with Theorem \ref{thm:2}.

\begin{figure}[h!]
    \centering
    \includegraphics[width=0.4\textwidth]{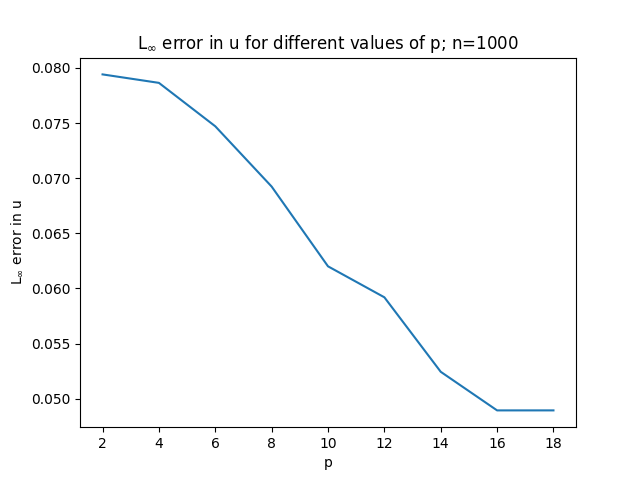} %
    \hfill
    \includegraphics[width=0.4\textwidth]{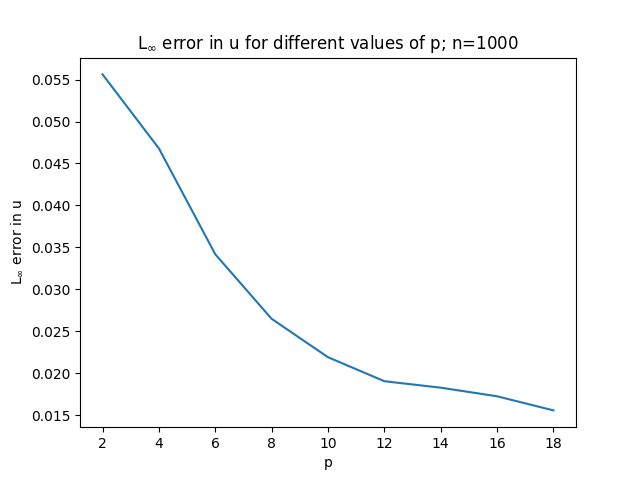} %
    \caption{Infinity norm error for varying number of columns ($\theta=0.1, \theta=0.4$)}
    \label{fig:inf_error}
\end{figure}

\newpage

\setcounter{section}{0}
\renewcommand{\thesection}{\Alph{section}}
\section{Appendix}

\subsection{Proof of Theorem 1}
\begin{proof}
    On trying all possible combinations of binary vectors for the columns of $\mathbf{X}$, we are effectively checking all cases.
    If the equations are consistent, the $\mathbf{u}$ vector obtained from the first column of $\mathbf{X}$ will match with the $\mathbf{u}$ vector obtained from the second column of $\mathbf{X}$.
    To identify the $\mathbf{u}$ vector, we need to solve the n simultaneous equations
    \begin{align*}
        \mathbf{Y_{ij}} & = \sum_{k=1}^{n} \mathbf{H_{ik}}\mathbf{X_{kj}}  \forall i \in [n]
    \end{align*}
    Using the definition for $\mathbf{H} \ (\mathbf{H_{ij}}=\delta_{ij}-2u_iu_j)$, we get
    \begin{equation}
        \begin{aligned}
            \mathbf{Y_{ij}} & = \sum_{k=1}^{n} \left(\delta_{ik}-2u_iu_k\right)\mathbf{X_{kj}} \ \forall i \in [n]
        \end{aligned}
        \label{eq:0}
    \end{equation}

    Here, $\delta_{ik}$ is the Kronecker delta function- its is $1$ for $i=k$ and $0$ otherwise. For any "guess" of the first column of $\mathbf{X}$, we can solve the above $n$ equations to get the $\mathbf{u}$ vector.
    We will get $2^n$ such $\mathbf{u}$ vectors for the first and the second column. If any of the $\mathbf{u}$ vectors obtained
    from the first column of $\mathbf{X}$ match with that of a $\mathbf{u}$ vector obtained from the second column of
    $\mathbf{X}$, then we have found the correct $\mathbf{u}$ vector. Our goal is to show that there will be atmost one vector
    $\mathbf{u}$ that works for both columns of $\mathbf{X}$. We prove this by contradiction. Assume that $(\mathbf{X_1}, \mathbf{u_1})$
    and $(\mathbf{X_2}, \mathbf{u_2})$ satisfy \eqref{eq:0} (let the corresponding Householder matrices be
    $(\mathbf{H_1}$ and $\mathbf{H_2})$ respectively). Define the sets (of non-supports) $S_{lj}=:\{k: \mathbf{X}_{lkj}=0\}$ for $j,l \in \{1,2\}$: this is the location of zeroes in the $j^{th}$ column of the $l^{th}$ solution.
    From \eqref{eq:0}, we have \(\sum_{ k \notin S_{1j}}(\delta_{ik}-2u_{1i}u_{1k}) = \sum_{k \notin S_{2j}} (\delta_{ik}-2u_{2i}u_{2k})  \).
    Define \(c_{lj} = \sum_{i \in S_{lj}} u_{li} \) for $l,j \in \{1,2\}$ as the sum of the entries in the $l^{th}$ solution for $\mathbf{u}$ at locations where the corresponding coefficient matrix is zero in the $j^{th}$ column. Likewise, let $c_{l} = \sum u_{li}$ be the sum of all entries in the $l^{th}$ solution for $\mathbf{u}$. (Thus, define $c_{s1}=c_{1j}$ and $c_{s2}=c_{2j}$).

    \begin{lemma}
        Following the notation from the previous paragraph, in case there are two solutions to \eqref{eq:0}, we can obtain the second solution from the first as
        \begin{align*}
            u_{2i} & = \frac{c_1-c_{s1}}{c_2-c_{s2}}u_{1i}  \ \ \ \hspace{73pt}  i \notin S_{1j},  i \notin S_{2j}                             \\
            u_{2i} & = \frac{c_1-c_{s1}}{c_2-c_{s2}}u_{1i} - \frac{1}{2}\left(\frac{1}{c_2-c_{s2}}\right)\quad i \notin S_{1j},\ i \in S_{2j}  \\
            u_{2i} & = \frac{c_1-c_{s1}}{c_2-c_{s2}}u_{1i} \ \ \ \hspace{73pt}  i \in S_{1j},\ i \in S_{2j}                                    \\
            u_{2i} & = \frac{c_1-c_{s1}}{c_2-c_{s2}}u_{1i} + \frac{1}{2}\left(\frac{1}{c_2-c_{s2}}\right) \quad i \in S_{1j},\ i \notin S_{2j}
        \end{align*}
        \label{lem:2}
    \end{lemma}

    \begin{proof}
        Now, we have 4 cases:
        \begin{enumerate}
            \item $i \notin S_{1j}$ and $i \notin S_{2j}$
            \item $i \notin S_{1j}$ and $i \in S_{2j}$
            \item $i \in S_{1j}$ and $i \in S_{2j}$
            \item $i \in S_{1j}$ and $i \notin S_{2j}$
        \end{enumerate}
        Define the 4 sets as $P_1, P_2, P_3, P_4$ respectively. We will now show that the $u_{2i}$'s can be expressed in terms of the $u_{1i}$'s for each case.
        Here, $i \in [n]$ indictaes the $i^{th}$ entry of the $\mathbf{u}$ vector,
        corresponding to the $i^{th}$ row of $\mathbf{X}$. This leads
        to the following equality constraints:
        \begin{align*}
            \sum_{i \notin S_{1j}} \left(\delta_{ik}-2u_{1i}u_{1k}\right) & = \sum_{i \notin S_{2j}} \left(\delta_{ik}-2u_{2i}u_{2k}\right) \\
            \sum_{i \notin S_{1j}} \left(\delta_{ik}-2u_{1i}u_{1k}\right) & = \sum_{i \in S_{2j}} \left(\delta_{ik}-2u_{2i}u_{2k}\right)    \\
            \sum_{i \in S_{1j}} \left(\delta_{ik}-2u_{1i}u_{1k}\right)    & = \sum_{i \in S_{2j}} \left(\delta_{ik}-2u_{2i}u_{2k}\right)    \\
            \sum_{i \in S_{1j}} \left(\delta_{ik}-2u_{1i}u_{1k}\right)    & = \sum_{i \notin S_{2j}} \left(\delta_{ik}-2u_{2i}u_{2k}\right)
        \end{align*}
        (Note that the summations are over $k$). These can be simplified as:
        \begin{align*}
            1-2u_{1i}(c_1-c_{s1}) & = 1-2u_{2i}(c_2-c_{s2}) \ \ \ i \notin S_{1j} \ and \ i \notin S_{2j} \\
            1-2u_{1i}(c_1-c_{s1}) & = -2u_{2i}(c_2-c_{s2}) \ \ \ i \notin S_{1j} \ and \ i \in S_{2j}     \\
            -2u_{1i}(c_1-c_{s1})  & = -2u_{2i}(c_2-c_{s2}) \ \ \ i \in S_{1j} \ and \ i \in S_{2j}        \\
            -2u_{1i}(c_1-c_{s1})  & = 1-2u_{2i}(c_2-c_{s2}) \ \ \ i \in S_{1j} \ and \ i \notin S_{2j}
        \end{align*}
        On rearranging the terms, we get the required result. We assume that we do not divide by 0 in any case.
    \end{proof}

    \begin{lemma}
        Following the notation above,
        \begin{enumerate}
            \item $\lvert P_2 \rvert = \lvert P_4 \rvert$,
            \item        \( (c_2-c_{s2})/(c_1-c_{s1}) = c_1/c_2 \)
            \item \(u_{1i}/u_{2i} = c_1/c_2 \) for  $i\in P_1\cup P_3$
            \item \(\sum_{i \in P_2 \cup P_4} u_{1i}/\sum_{i \in P_2 \cup P_4} u_{2i}  = c_1/c_2 \)
        \end{enumerate}
        \label{lem:3}
    \end{lemma}

    \begin{proof}
        First, we use the fact that any solution to $\mathbf{u}$ must be unit norm.
        \begin{equation}
            \begin{aligned}
                \sum_{i=1}^{n} u_{2i}^2 & = 1 \label{eq:2} \\
            \end{aligned}
        \end{equation}
        Substituting the equations from lemma \ref{lem:2} into \eqref{eq:2}, we get
        \begin{align*}
            \sum_{i \notin S_{1j} \ and \ i \notin S_{2j}} \left(\frac{c_1-c_{s1}}{c_2-c_{s2}}u_{1i}\right)^2 +
            \sum_{i \notin S_{1j} \ and \ i \in S_{2j}} \left(\frac{c_1-c_{s1}}{c_2-c_{s2}}u_{1i} - \frac{1}{2}\left(\frac{1}{c_2-c_{s2}}\right)\right)^2  +    \\
            \sum_{i \in S_{1j} \ and \ i \in S_{2j}} \left(\frac{c_1-c_{s1}}{c_2-c_{s2}}u_{1i}\right)^2                                                    +
            \sum_{i \in S_{1j} \ and \ i \notin S_{2j}} \left(\frac{c_1-c_{s1}}{c_2-c_{s2}}u_{1i} + \frac{1}{2}\left(\frac{1}{c_2-c_{s2}}\right)\right)^2 & = 1 \\
        \end{align*}
        On grouping the terms appropriately, we get
        \begin{align*}
            \left(\frac{c_1-c_{s1}}{c_2-c_{s2}}\right)^2 \left(\sum_{P_1 \cup P_2 \cup P_3 \cup P_4} u_{1i}^2\right) + \sum_{P_2 \cup P_4}\frac{1}{4}\left(\frac{1}{c_2-c_{s2}} \right)^2
            -                                                                                                                                     \\
            \sum_{P_2}\left(\frac{c_1-c_{s1}}{(c_2-c_{s2})^2}\right)u_{1i} + \sum_{P_4}\left(\frac{c_1-c_{s1}}{(c_2-c_{s2})^2}\right)u_{1i} & = 1
        \end{align*}
        On using the unit norm condition on $\mathbf{u_1}$:
        \begin{align*}
            \left(\frac{c_1-c_{s1}}{c_2-c_{s2}}\right)^2 + \sum_{P_2 \cup P_4}\frac{1}{4}\left(\frac{1}{c_2-c_{s2}} \right)^2
            - \sum_{P_2}\left(\frac{c_1-c_{s1}}{(c_2-c_{s2})^2}\right)u_{1i} + \sum_{P_4}\left(\frac{c_1-c_{s1}}{(c_2-c_{s2})^2}\right)u_{1i} & = 1
        \end{align*}
        On cross multiplying and expanding:
        \begin{align*}
            (c_1-c_{s1})^2 + \frac{\lvert P_2 \rvert + \lvert P_4 \rvert}{4}
            - \sum_{P_2}\left(c_1-c_{s1}\right)u_{1i} + \sum_{P_4}\left(c_1-c_{s1}\right)u_{1i} & = (c_2-c_{s2})^2
        \end{align*}
        Which on simplifying leads to:
        \begin{equation}
            \begin{aligned}
                (c_2-c_{s2})^2 - (c_1-c_{s1})^2 & = \frac{\lvert P_2 \rvert + \lvert P_4 \rvert}{4}-
                \sum_{P_2}\left(c_1-c_{s1}\right)u_{1i} + \sum_{P_4}\left(c_1-c_{s1}\right)u_{1i} \label{eq:3}
            \end{aligned}
        \end{equation}

        \begin{enumerate}
            \item \begin{proof}
                      Consider the equation obtained from the sums of entries of the
                      $\mathbf{u}$ vectors.
                      \begin{align*}
                          \sum_{i=1}^{n} u_{2i} & = c_2
                      \end{align*}
                      Plugging in the expressions from Lemma \ref{lem:2}, we get
                      \begin{align*}
                          \sum_{i \notin S_{1j} \ and \ i \notin S_{2j}} \left(\frac{c_1-c_{s1}}{c_2-c_{s2}}u_{1i}\right) +
                          \sum_{i \notin S_{1j} \ and \ i \in S_{2j}} \left(\frac{c_1-c_{s1}}{c_2-c_{s2}}u_{1i} - \frac{1}{2}\left(\frac{1}{c_2-c_{s2}}\right)\right)  +      \\
                          \sum_{i \in S_{1j} \ and \ i \in S_{2j}} \left(\frac{c_1-c_{s1}}{c_2-c_{s2}}u_{1i}\right)                                                    +
                          \sum_{i \in S_{1j} \ and \ i \notin S_{2j}} \left(\frac{c_1-c_{s1}}{c_2-c_{s2}}u_{1i} + \frac{1}{2}\left(\frac{1}{c_2-c_{s2}}\right)\right) & = c_2 \\
                      \end{align*}
                      Thus,
                      \begin{align*}
                          \left(\frac{c_1-c_{s1}}{c_2-c_{s2}}\right) \left(\sum_{P_1 \cup P_2 \cup P_3 \cup P_4} u_{1i}\right)
                          - \sum_{P_2}\frac{1}{2}\left(\frac{1}{c_2-c_{s2}}\right) + \sum_{P_4}\frac{1}{2}\left(\frac{1}{c_2-c_{s2}}\right) & = c_2
                      \end{align*}
                      This gives
                      \begin{align*}
                          c_1\left(\frac{c_1-c_{s1}}{c_2-c_{s2}}\right)
                          - \sum_{P_2}\frac{1}{2}\left(\frac{1}{c_2-c_{s2}}\right) + \sum_{P_4}\frac{1}{2}\left(\frac{1}{c_2-c_{s2}}\right) & = c_2
                      \end{align*}
                      On cross multiplying, we get
                      \begin{align*}
                          c_1\left(c_1-c_{s1} \right)
                          - \sum_{P_2}\frac{1}{2} + \sum_{P_4}\frac{1}{2} & = c_2(c_2-c_{s2})
                      \end{align*}
                      On simplifying, we have,
                      \begin{equation}
                          \begin{aligned}
                              \implies - \sum_{P_2}\frac{1}{2} +
                              \sum_{P_4}\frac{1}{2} & = c_2(c_2-c_{s2}) - c_1\left(c_1-c_{s1}\right) \label{eq:4}
                          \end{aligned}
                      \end{equation}
                      We then consider sums over $P2$ and $P3$:

                      \begin{align*}
                          \sum_{P2 \cup P3} u_{2i} = c_{s2}
                      \end{align*}
                      This gives
                      \begin{align*}
                          \sum_{P2}u_{2i} + \sum_{P3}u_{2i} & = c_{s2}
                      \end{align*}
                      Substituting the expressions from Lemma \ref{lem:2}, we get
                      \begin{align*}
                          \sum_{P2 \cup P3} \left(\frac{c_1-c_{s1}}{c_2-c_{s2}}\right)u_{1i}
                          - \sum_{P_2}\frac{1}{2}\left(\frac{1}{c_2-c_{s2}}\right) & = c_{s2}
                      \end{align*}
                      On cross-multiplying and rearranging,
                      \begin{equation}
                          \begin{aligned}
                              \sum_{P2 \cup P3} \left(c_1-c_{s1}\right)u_{1i}
                              - \sum_{P_2}\frac{1}{2} & = c_{s2}(c_2-c_{s2}) \label{eq:5}
                          \end{aligned}
                      \end{equation}
                      Correpondingly, consider the sums over $P3$ and $P4$:
                      \begin{align*}
                          \sum_{P3 \cup P4} u_{1i} = c_{s1}
                      \end{align*}
                      Decomposing similar to the previous case,
                      \begin{align*}
                          \sum_{P2}u_{1i} + \sum_{P3}u_{1i} = c_{s1}
                      \end{align*}
                      Substituting the expressions from Lemma \ref{lem:2},
                      \begin{align*}
                          \sum_{P3 \cup P4} \left(\frac{c_2-c_{s2}}{c_1-c_{s1}}\right)u_{2i}
                          - \sum_{P_4}\frac{1}{2}\left(\frac{1}{c_1-c_{s1}}\right) & = c_{s1}
                      \end{align*}
                      Simplifying and rearranging,
                      \begin{equation}
                          \begin{aligned}
                              \sum_{P3 \cup P4} \left(c_2-c_{s2}\right)u_{2i}
                              - \sum_{P_4}\frac{1}{2} & = c_{s1}(c_1-c_{s1}) \label{eq:6}
                          \end{aligned}
                      \end{equation}

                      Thus, we get (from $\ref{eq:5} \ - \ \ref{eq:6}$):
                      \begin{equation}
                          \begin{aligned}
                              c_{s2}(c_2-c_{s2}) - c_{s1}\left(c_1-c_{s1}\right) & =  \sum_{P2} (c_1-c_{s1})u_{1i}
                              -  \sum_{P4} (c_2-c_{s2})u_{2i} - \sum_{P_2}\frac{1}{2} +
                              \sum_{P_4}\frac{1}{2} \label{eq:7}
                          \end{aligned}
                      \end{equation}

                      On subtracting from the previous equation ($\ref{eq:4} \ - \ \ref{eq:7}$):
                      \begin{align*}
                          c_2(c_2-c_{s2}) - c_1\left(c_1-c_{s1}\right)-(c_{s2}(c_2-c_{s2}) - c_{s1}\left(c_1-c_{s1}\right))
                           & = - \left(\sum_{P2} (c_1-c_{s1})u_{1i} -  \sum_{P4} (c_2-c_{s2})u_{2i}\right)
                      \end{align*}
                      This gives us:
                      \begin{equation}
                          \begin{aligned}
                              (c_2-c_{s2})^2-(c_1-c_{s1})^2 & = - \sum_{P2} (c_1-c_{s1})u_{1i} +  \sum_{P4} (c_2-c_{s2})u_{2i} \label{eq:8}
                          \end{aligned}
                      \end{equation}

                      Comparing with $\ref{eq:3}$, we get:
                      \begin{align*}
                          \frac{\lvert P_2 \rvert + \lvert P_4 \rvert}{4} & = \sum_{P4} (c_2-c_{s2})u_{2i} - \sum_{P4} (c_1-c_{s1})u_{1i}
                      \end{align*}

                      \begin{align*}
                          \implies \lvert P_2 \rvert + \lvert P_4 \rvert & = 2\sum_{P4} -2u_{1i}(c_1-c_{s1})
                          + 2u_{2i}(c_2-c_{s2})                                                                                   \\
                          \implies \lvert P_2 \rvert + \lvert P_4 \rvert & = 2\sum_{P4} 1-2u_{2i}(c_2-c_{s2})+2u_{2i}(c_2-c_{s2}) \\
                          \implies \lvert P_2 \rvert + \lvert P_4 \rvert & = 2\sum_{P4} 1
                      \end{align*}
                      Thus, we have
                      \begin{equation}
                          \begin{aligned}
                              \lvert P_2 \rvert - \lvert P_4 \rvert & =  0 \label{eq:9}
                          \end{aligned}
                      \end{equation}

                      And this concludes the proof of the first part of the lemma. Note that this could have been obtained directly
                      by taking the $l_2$ norm (column wise) on both sides of the equation $\mathbf{H_1X_1=H_2X_2}$. However, the following results
                      are not as straightforward to obtain without the analysis of the underlying structure of the problem.
                  \end{proof}

            \item \begin{proof}
                      Using $\ref{eq:9} \ $ in $ \ \ref{eq:4}$, we get:
                      \begin{align*}
                          0 & = c_2(c_2-c_{s2}) - c_1\left(c_1-c_{s1}\right)
                      \end{align*}
                      On rearranging the terms, we get
                      \begin{equation}
                          \begin{aligned}
                              \frac{c_1}{c_2} & = \frac{c_2-c_{s2}}{c_1-c_{s1}} \label{eq:10}
                          \end{aligned}
                      \end{equation}

                  \end{proof}

            \item \begin{proof}
                      Using lemma $\ref{lem:2}$, and substituting the result 2. from lemma $\ref{lem:3}$, we get:
                      \begin{equation}
                          \begin{aligned}
                              \frac{u_{1i}}{u_{2i}} & = \frac{c_1}{c_2} \ \ \ \text{for}  \ \ \ P_1 \cup P_3 \label{eq:13}
                          \end{aligned}
                      \end{equation}
                  \end{proof}

            \item \begin{proof}
                      Using $\ref{eq:9} \ $ in $ \ \ref{eq:7}$, we get:
                      \begin{align*}
                          c_{s2}(c_2-c_{s2}) - c_{s1}\left(c_1-c_{s1}\right) & =  \sum_{P2} (c_1-c_{s1})u_{1i} - \sum_{P4} (c_2-c_{s2})u_{2i}
                      \end{align*}
                      Dividing the equation by $c_1-c_{s1}$, we get
                      \begin{align*}
                          c_{s2}\frac{c_2-c_{s2}}{c_1-c_{s1}}-c_{s1} & =  \sum_{P2} u_{1i} - \sum_{P4} \frac{c_2-c_{s2}}{c_1-c_{s1}}u_{2i}
                      \end{align*}
                      Using result 2. from lemma $\ref{lem:3}$, we get:
                      \begin{align*}
                          c_{s2}\frac{c_1}{c_2}-c_{s1} & = \sum_{P2} u_{1i} - \sum_{P4} \frac{c_1}{c_2}u_{2i}
                      \end{align*}
                      Thus, we have:
                      \begin{align*}
                          \frac{c_1}{c_2}\sum_{P2 \cup P3 \cup P4} u_{2i} & = \sum_{P2 \cup P3 \cup P4} u_{1i}
                      \end{align*}
                      Or equivalently,
                      \begin{equation}
                          \begin{aligned}
                              \frac{c_1}{c_2} & = \frac{\sum_{P2 \cup P3 \cup P4} u_{1i}}{\sum_{P2 \cup P3 \cup P4} u_{2i}} \label{eq:11}
                          \end{aligned}
                      \end{equation}

                      Using $\ref{eq:13} \ $ in $ \ \ref{eq:11}$:
                      \begin{align*}
                          \sum_{P2 \cup P3 \cup P4} u_{1i} & = \sum_{P2 \cup P3 \cup P4} u_{2i} \frac{c_1}{c_2} \\
                      \end{align*}
                      Separating the terms, we get:
                      \begin{align*}
                          \implies \sum_{P2 \cup P4} u_{1i} + \sum_{P3} u_{1i} & = \sum_{P2 \cup P4} u_{2i}\frac{c_1}{c_2} + \sum_{P3} u_{2i} \frac{c_1}{c_2} \\
                      \end{align*}
                      Using result 3. from lemma $\ref{lem:3}$, we get:
                      \begin{align*}
                          \sum_{P2 \cup P4} u_{1i} & = \sum_{P2 \cup P4} u_{2i} \frac{c_1}{c_2}
                      \end{align*}
                      On simplifying, we have the result:
                      \begin{equation}
                          \begin{aligned}
                              \frac{\sum_{P2 \cup P4} u_{1i}}{\sum_{P2 \cup P4} u_{2i}} & = \frac{c_1}{c_2}
                          \end{aligned}
                      \end{equation}
                  \end{proof}
        \end{enumerate}

    \end{proof}

    Thus, the number of row indices for which the index doesn't belong to $S_{1j}$ but
    belongs to $S_{2j}$ is equal to the number of row indices for which the index belongs to $S_{1j}$ but
    doesn't belong to $S_{2j}$. What does the index not belonging to
    $S_{1j}$ but belonging to $S_{2j}$ mean? It means that the corresponding entry in the $\mathbf{X_1}$ matrix is non-zero for the column
    under consideration (i.e., it is 1) whilst the corresponding entries in the $\mathbf{X_2}$ matrix are zero. Similarly, the index belonging to
    $S_{1j}$ but not belonging to $S_{2j}$ means that the corresponding entry in the $\mathbf{X_2}$ matrix is non-zero for the column
    under consideration (i.e., it is 1) whilst the corresponding entries in the $\mathbf{X_1}$ matrix are zero. Thus, the number of non-zero entries
    are the same in both matrices $\mathbf{X_1}$ and $\mathbf{X_2}$ for a given column
    (since the rest of the entries are either both 1 or both 0). Thus, the two matrices are only
    different due to permutation. \\[10pt]
    Now, we consider multiple columns. We know that every
    column of $\mathbf{X}$ has a permutation of 0's and 1's corresponding to the
    ground truth $\mathbf{u}$ vector. The corresponding column vector in
    $\mathbf{Y}$ is different for different columns. However, the ground truth $\mathbf{u}$ vector
    remains the same throughout. Thus, the $\mathbf{u}$ vector generated from one of the permutations
    will certainly match that generated in the previous columns. An error is caused
    when the same "incorrect" $\mathbf{u}$ vector is generated for every column. \\
    Consider Equation $\ref{eq:3}$. Say we know the ground truth $\mathbf{u}$ vector. From $\ref{eq:3} \ $, we have:

    \begin{align*}
        (c_2-c_{s2})^2 - (c_1-c_{s1})^2 & = \frac{\lvert P_2 \rvert + \lvert P_4 \rvert}{4}-
        \sum_{P_2}\left(c_1-c_{s1}\right)u_{1i} + \sum_{P_4}\left(c_1-c_{s1}\right)u_{1i}
    \end{align*}
    Thus,
    \begin{equation}
        \begin{aligned}
            (c_2-c_{s2}) & = \pm \sqrt{(c_1-c_{s1})^2+(c_1-c_{s1})\left(\sum_{P_4}u_{1i} -
            \sum_{P_2}u_{1i}\right)+\frac{\lvert P_2 \rvert + \lvert P_4 \rvert}{4}} \label{eq:15}
        \end{aligned}
    \end{equation}

    For this value to always be defined, the term inside the square root must always
    be non-negative. Thus,
    \begin{align*}
        (c_1-c_{s1})^2+(c_1-c_{s1})\left(\sum_{P_4}u_{1i} -
        \sum_{P_2}u_{1i}\right)+\frac{\lvert P_2 \rvert + \lvert P_4 \rvert}{4} \geq 0 \ \forall \mathbf{u_1}
    \end{align*}

    Hence, the discriminant of the above quadratic must always be non-positive. Let's analyze the discriminant:
    \begin{align*}
        \Delta & = \left(\sum_{P_4}u_{1i} - \sum_{P_2}u_{1i}\right)^2 - 4\left(\frac{\lvert P_2 \rvert + \lvert P_4 \rvert}{4}\right) \\
    \end{align*}
    Or,
    \begin{align*}
        \Delta & = \left(\sum_{P_4}u_{1i} - \sum_{P_2}u_{1i}\right)^2 - \lvert P_2 \rvert - \lvert P_4 \rvert
    \end{align*}
    Since
    \begin{align*}
        \max \left(\sum_{P_4}u_{1i} - \sum_{P_2}u_{1i}\right) & = \sqrt{\lvert P_2 \rvert + \lvert P_4 \rvert}
    \end{align*}

    \begin{align*}
        \Delta & \leq 0 \ \forall \mathbf{u_1}
    \end{align*}

    Thus, when given a $\mathbf{u_1}$ vector, the corresponding possible $\mathbf{u_2}$ vector can be found.
    Only 2 of these can exist for a given ($\mathbf{u_1}$ vector, column vector of $\mathbf{X_1}$) pair and only one can exist upto sign.
    Now, with respect to the ground truth column vector $\mathbf{X_1}$, which contains say $k$ ones, we want
    to find how many other $\mathbf{u_2}$, $\mathbf{X_2}$ pairs can give the same column vector in $\mathbf{Y}$, upto
    sign. From $\ref{eq:15}$, we conclude that for a given $\mathbf{u_1}$ vector, the $\mathbf{u_2}$ vector depends on $\lvert P_4 \rvert$
    (or equivalently $\lvert P_2 \rvert$) only. This number can vary from $0$ to at most $k$. $k$ must be less than
    or equal to $\frac{n}{2}$, since at most half the entries can be of the same kind whilst still
    ensuring $\lvert P_2 \rvert = \lvert P_4 \rvert$. Thus, the total number of
    $\mathbf{H}, \mathbf{X}$ pairs that can give the same $\mathbf{Y}$ from the first column is at most $k=O(n)$
    (subject to sign and permutation). However, we argue that this does not affect the number of columns
    needed to recover the ground truth $\mathbf{u}$ vector. \\
    To do so, we now consider the
    second column. We need to check if the same $\mathbf{u_2}$ vector obtained from the first column can satisfy the equations
    corresponding to the second column. Note that there is a conflict
    iff a $\mathbf{u_2}$ vector generated from the second column is
    exactly the same as that generated from the first column. This is because,
    when we solve the equations to find the ground truth $\mathbf{u_1}$ vector, the result is exactly
    the same from both columns, since this is how the matrix $\mathbf{Y}$ was generated.
    There is no variation even in permutation or sign. Thus, we check if this is possible. Assume that the
    $\mathbf{u_2}$ vector generated from the second column is the same as that generated from the first column. This implies that both
    $c_1$ and $c_2$ are the same.

    Consider the lemma $\ref{lem:2}$.

    Using $\ref{eq:10} \ $ and $ \ \ref{eq:13}$, the equations in lemma $\ref{lem:2}$ are equal to:

    \begin{align*}
        u_{1i} & = u_{2i}\frac{c_1}{c_2} \ \ \ i \in P_1 \cup P_3                 \\
        u_{1i} & = \frac{1}{2(c_1-c_{s1})}+u_{2i}\frac{c_1}{c_2} \ \ \ i \in P_2  \\
        u_{1i} & = -\frac{1}{2(c_1-c_{s1})}+u_{2i}\frac{c_1}{c_2} \ \ \ i \in P_4 \\
    \end{align*}

    Note that since $\mathbf{u_1}$ is the same, $c_1$ is the same in the equations above as
    that from the first column. Thus, if $\mathbf{u_2}$ has to be the same
    , then $c_1-c_{s1}$ must be the same for the ground $\mathbf{u_1}$ vector generated from the first column.
    Thus, $c_{s1}$ must also be the same. Not only this, but $\lvert P_2 \rvert$ and
    $\lvert P_4 \rvert$ must be the same because there is a one-to-one relationship
    between the known ground truth $\mathbf{u_1}$ vector and the $\mathbf{u_2}$ vector generated from any column.
    If $\lvert P_2 \rvert $ is different, the number of entries which are calculated using the above equation
    for $i \in P_2$ will change. Thus, the $\mathbf{u_2}$ vector generated will be different. The only way it won't change is
    if $\frac{1}{c_1-c_{s1}}$ is $0$. But this is not possible.
    Therefore, the only way of getting an identical $\mathbf{u_2}$ vector is if the ground truth $\mathbf{X}$ column vector for the
    first and second column are identical. But this is a contradiction, since we assume them to be different. Thus, we will
    be able to recover the ground truth $\mathbf{u_1}$ vector uniquely, by using only 2 columns.

\end{proof}

\subsection{Proof of Theorem 2}

\begin{proof}
    We use a proof by counter-example here. We will show that there exists a pair of matrices,
    $(\mathbf{H_1}, \mathbf{X_1})$ and another pair $(\mathbf{H_2}, \mathbf{X_2})$, such that
    $\mathbf{H_1}\mathbf{X_1} = \mathbf{H_2}\mathbf{X_2}$. Consider the householder vectors
    \begin{align*}
        \mathbf{u_1} & = \begin{bmatrix}
                             \sqrt{\frac{1}{3}} \\
                             \sqrt{\frac{2}{3}} \\
                         \end{bmatrix}    \\
        \mathbf{u_2} & = \begin{bmatrix}
                             \frac{1}{\sqrt{2}} \\
                             \frac{1}{\sqrt{2}}
                         \end{bmatrix} \\
    \end{align*}
    The corresponding householder matrices are
    \begin{align*}
        \mathbf{H_1} & = \begin{bmatrix}
                             \frac{1}{3}          & \frac{-2\sqrt{2}}{3} \\
                             \frac{-2\sqrt{2}}{3} & \frac{-1}{3}
                         \end{bmatrix} \\
        \mathbf{H_2} & = \begin{bmatrix}
                             0  & -1 \\
                             -1 & 0
                         \end{bmatrix}                                             \\
    \end{align*}
    Consider the $p^{th}$ column vector of $\mathbf{X_1}$ and $\mathbf{X_2}$ as
    \begin{align*}
        \mathbf{X_{1p}} & = \begin{bmatrix}
                                x_{11p} \\
                                x_{12p}
                            \end{bmatrix} \\
        \mathbf{X_{2p}} & = \begin{bmatrix}
                                x_{21p} \\
                                x_{22p}
                            \end{bmatrix} \\
    \end{align*}
    Thus, the corresponding column vectors of the $\mathbf{Y}$ matrices are
    \begin{align*}
        \mathbf{Y_{1p}} & = \begin{bmatrix}
                                \frac{x_{11p}-2\sqrt{2}x_{12p}}{3} \\
                                \frac{-x_{12p}-2\sqrt{2}x_{11p}}{3}
                            \end{bmatrix} \\
        \mathbf{Y_{2p}} & = \begin{bmatrix}
                                -x_{22p} \\
                                -x_{21p}
                            \end{bmatrix}                     \\
    \end{align*}
    We need $\mathbf{Y_{1p}} = \mathbf{Y_{2p}}$. Thus, we need
    \begin{align*}
        \frac{x_{11p}-2\sqrt{2}x_{12p}}{3}  & = -x_{22p} \\
        \frac{-x_{12p}-2\sqrt{2}x_{11p}}{3} & = -x_{21p} \\
    \end{align*}
    This can be done for every column $p$. For example, a
    satisfying assignment would be:

    \begin{align*}
        \mathbf{X_{1p}} & = \begin{bmatrix}
                                \frac{2\sqrt2}{3} \\
                                \frac{1}{3}
                            \end{bmatrix}
    \end{align*}

    \begin{align*}
        \mathbf{X_{2p}} & = \begin{bmatrix}
                                1 \\
                                0
                            \end{bmatrix}
    \end{align*}

    Thus, the solution is not unique upto sign and permutation. Note that even for
    Theorem 1, multiple $\mathbf{H}, \mathbf{X}$ pairs could give the
    same $\mathbf{Y}$ for a given column. However, we rely on the fact that entries of
    the column vectors of $\mathbf{X}$ are constrained to be 0 or 1 only.
    This would ensure that upto permutation, only $O(n)$ such possibilities existed.
    However, without such a restriction here, infinite solutions are possible and thus recovery
    is not possible.

\end{proof}

\subsection{Proof of Lemma 1}

\begin{proof}
    Consider the matrices $\mathbf{Y} \in \mathbf{R^{n \times p}}$,
$\mathbf{V} \in \mathbf{R^{n \times n}}$, \ $\mathbf{X} \in \mathbf{R^{n \times p}}$,
such that $\mathbf{Y}=\mathbf{V}\mathbf{X}$

\begin{align*}
    \mathbf{Y_{ij}}                  & =\sum_{k=1}^{n} \mathbf{V_{ik}} \mathbf{X_{kj}}              \\
    \mathbb{E}[ \mathbf{X_{ij}} ]    & = \theta                                                     
\end{align*}
\begin{align*}
    \mathbf{Y^TY=(VX)^T(VX)}         & = \mathbf{X^TV^TVX} = \mathbf{X^TX} 
\end{align*}
\begin{align*}
    trace(\mathbf{Y^TY})             & = \sum_{i=1}^{n}\sum_{j=1}^{p} \mathbf{Y_{ij}}^2 =            \sum_{i=1}^{n}\sum_{j=1}^{p} \mathbf{X_{ij}}^2             \\
    \mathbb{E}(trace(\mathbf{Y^TY})) & = \sum_{i=1}^{n}\sum_{j=1}^{p} \mathbb{E}[\mathbf{X_{ij}}^2] \\
                                     & = \sum_{i=1}^{n}\sum_{j=1}^{p} \theta                        \\
                                     & = np\theta                                                   \\
\end{align*}

\noindent We set the expected value of the above sum of products as the actual value
of the above sum of products, and prove these values are close to each other, by using Hoeffding's inequality \cite{hoeffding1994probability}.

Define new random variables $Z_m$ as $\mathbf{X_{ij}}^2/{np}$ for $m=1,2,...,np$. Note that for Bernoulli random variables,
$\mathbf{X_{ij}}^2=\mathbf{X_{ij}}$. Since the $\mathbf{X_{ij}}$ are independent, the $Z_m$ are also independent. Thus, we can apply Hoeffding's inequality to the sum of these random variables. 
Define $S = \sum_{m=1}^{np} Z_m$. Given each $Z_m$ is bounded as follows: $a_m \leq Z_m \leq b_m$: from Hoeffding's inequality, we have:
\begin{align*}
    \mathbb{P} \left [\lvert  S -\mathbb{E}[S] > t \rvert \right ] \leq 2exp\left(-2t^2 / \sum_{m=1}^{np} (a_m-b_m)^2\right)
\end{align*}

Here, $a_m=0$ and $b_m=1/{np}$. Thus, $\sum_{m=1}^{np} (a_m-b_m)^2 = \sum_{m=1}^{np} 1/{np}^2 = 1/{np}$.
Finally, according to Hoeffding's inequality, :

$$
    \mathbb{P}\left[ \left \lvert \frac{\sum_{i=1}^{n}\sum_{j=1}^{p} \mathbf{Y_{ij}}^2}
        {(n)(p)} - \theta \right \rvert > t \right]
    \leq 2 \exp \left(-2t^2 \ (n)(p) \right)
$$
\end{proof}

\subsection{Proof of Theorem 3}

\begin{proof}
    Consider the matrices $\mathbf{Y} \in \mathbf{R^{n \times p}}$,
$\mathbf{H} \in \mathbf{R^{n \times n}}$, \ $\mathbf{X} \in \mathbf{R^{n \times p}}$,
such that $\mathbf{Y}=\mathbf{H}\mathbf{X}$

\begin{align*}
    \mathbf{Y_{ij}}                  & =\sum_{k=1}^{n} \mathbf{H_{ik}} \mathbf{X_{kj}}              \\
    &= \sum_{k=1}^{n}(\delta_{ik}-2u_iu_k)\mathbf{X_{kj}} 
\end{align*}

    Note that
\begin{align*}
        \sum_{i=1}^{n}\sum_{j=1}^{p}\mathbb{E}(\mathbf{Y_{ij}}) & = \sum_{i=1}^{n}\sum_{j=1}^{p}\sum_{k=1}^{n}\mathbb{E}((\delta_{ik}-2u_iu_k)\mathbf{X_{kj}}) \\
        &= \sum_{i=1}^{n}\sum_{k=1}^{n} p \theta (\delta_{ik}-2u_iu_k) \\
        &= p\theta \sum_{i=1}^{n}(1-2u_ic) \\
        &= p\theta(n-2c^2)
    \end{align*}

where $c=\sum_{i=1}^{n} u_i$. Now, define
\[
 k_i = \frac{-1}{2}\left( \frac{\sum_{j = 1 }^{p} \mathbf{(Y)_{ij}}}{p\theta} - 1 \right)
= u_{i}\left( \sum_{z=1}^{n} u_z \right),     
\]
From this, we can take the ratio of terms to get 
\[\frac{k_z}{k_i} = \frac{u_z\left( \sum_{z=1}^{n} u_z \right)}{u_i\left( \sum_{z=1}^{n} u_z \right)}= \frac{u_z}{u_i}\] 
Thus, we have \[k_i = u_{i}\left( \sum_{z=1}^{n} u_i \frac{k_z}{k_i} \right) \]. Thus, 
\[k_i^2  = u_{i}^2\left( \sum_{z=1}^{n} k_z \right)\]
 which gives \[u_{i}  = \frac{k_i}{\sqrt{\sum_{z=1}^{n} k_z}}\]
        
For each $u_{1i}$, we take the sign of each $u_{1i}$ to be the same as that of $k_i$. 
        The computational complexity involved in the calculation of $\mathbf{u}$ is $O(np)$. We get the desired result by applying Hoeffding's inequality, once for the calculation of $c$ and once for the calculation of the $u_i$'s. \\
        Note that we find the variation from $c/\sqrt{n}$ and use this in finding a bound for the overall 
        error. 
        Hoeffding's Inequality for $c$ calculation:
        Define $Z_m$ as $\mathbf((-1/2)({Y_{ij}}-\theta)/{np\theta})$ for $m=1,2,...,np$. This can be written as 
        \[ Z_m = \frac{-1}{2}\left ( \frac{\left ( \sum_{k=1}^{n}(\delta_{ik}-2u_iu_k)\mathbf{X_{kj}} \right  )-\theta}{np\theta} \right ) \]
        Define $S = \sum_{m=1}^{np} Z_m$ 
        \[\mathbb{E}[S] = \mathbb{E} \left [\sum_{m=1}^{np}Z_m \right ] = \frac{c^2}{n} = \left (\frac{c}{\sqrt{n}} \right )^2  \]
       Again, similar to lemma \ref{lem:1}, we have: $a_m = 1/{2np}$ and $b_m = (-1/2)(1-2u_ic)+1/{2np}$. Thus, $a_m-b_m=(-1/2)(1-2u_ic)$. Finally
       $\sum_{m=1}^{np}(a_m-b_m)^2=\sum_{m=1}^{np}(1/4)(1-4u_ic+4u_i^2c^2)=np/4$. Hence, according to Hoeffding's inequality:
        \begin{align*}
    \mathbb{P}\left[ \left\lvert \left(\left(\frac{\sum_{i=1}^{n} \sum_{j=1}^{p} \mathbf{Y_{ij}}}{np\theta}\right)-1\right)\frac{-1}{2} - \left(\frac{c}{\sqrt{n}}\right)^2 \right\rvert > t \right] \leq 2 \exp\left( -8t^2  \theta^2 np \right) 
\end{align*}

Finally, we analyze the calculation of the $u_i$'s.
We can either set the $u_i$'s based on Algorithm \ref{find_V_X} or we can directly set the $u_i$'s based on the following:

Note that both the method explained in Algorithm \ref{find_V_X} and what we have described here are equivalent.
They represent two different ways of expressing the same quantity. \\
Set \[u_i= \left (\frac{\sum_{j=1}^{p} {\mathbf{Y_{ij}}}}{p\theta}-1 \right ) \left(\frac{-1}{2c} \right) \]
Set $c' = c/\sqrt{n}$. Then we have:
\[u_i= \left (\frac{\sum_{j=1}^{p} {\mathbf{Y_{ij}}}}{p\theta}-1 \right )\left(\frac{-1}{2c'\sqrt{n}}\right)=y\left(\frac{-1}{2c'\sqrt{n}}\right)\]
As a consequence, the error in any given \(u_i\) has two components - due to the error in $c'$ and the deviation in the empirical correlation above. We bound this by taking the union bound over the probability of error in each case. For \(p=\Omega(1)\), the error in \(c'\) is low. Thus, the governing term is the error in the first term.
Just like the previous proofs, we once again define a random variable $Z_m = (-1/2c)(\mathbf{Y_{ij}}-\theta)/{p\theta}$. Effectively,
\[ Z_m = \frac{-1}{2c}\left ( \frac{\left ( \sum_{k=1}^{n}(\delta_{ik}-2u_iu_k)\mathbf{X_{kj}} \right  )-\theta}{p\theta} \right )\]
Define $S=\sum_{m=1}^{p} Z_m$. 
\[\mathbb{E}[S]=\mathbb{E}\left [\sum_{m=1}^{p} Z_m \right ]=u_i\]

Similar to above, we have $a_m = 1/{2cp}$ and $b_m = (-1/2c)(\delta_{im}-2u_iu_m)+1/{2cp}$. Thus, $a_m-b_m=(-1/2)(\delta_{im}-2u_iu_m)$.
$\sum_{m=1}^{p}(a_m-b_m)^2=\sum_{m=1}^{p}(1/4c^2)(\delta_{im}^2-4u_iu_m\delta_{im}+4u_i^2u_m^2)=p/(4c^2)$ Finally
According to Hoeffding's inequality:

\[
\mathbb{P}\left[\left \lvert \left(\frac{\sum_{j=1}^{p} \mathbf{Y_{ij}}}{p\theta}-1 \right)\left(\frac{-1}{2c}\right) - u_i \geq t \right \rvert \right] \leq 2\exp(-8t^2c^2\theta^2p) 
\]

This is for any given \(u_i\). We take the union bound over all \(u_i\)'s since we want to find the error in the $l_{\infty}$ sense.
We want this error to be (over the union bound), at most, say, $O(\frac{1}{n})$. Hence, \(2nexp(-8t^2c^2\theta^2p) \leq \frac{1}{n}\). On solving, we get
\begin{align*}
    p = \Omega\left(\frac{\log(2n^2)}{8t^2c^2\theta^2}\right)
\end{align*}

\end{proof}


\begin{thebibliography}{99}

    \bibitem{sun2016complete}
    Sun, J., Qu, Q., \& Wright, J. (2016).
    \textit{Complete dictionary recovery over the sphere I: Overview and the geometric picture}.
    IEEE Transactions on Information Theory, 63(2), 853--884.

    \bibitem{yeganli2014improved}
    Yeganli, F., \& Ozkaramanli, H. (2014).
    \textit{Improved online dictionary learning for sparse signal representation}.
    In 2014 22nd Signal Processing and Communications Applications Conference (SIU) (pp. 1702--1705).
    IEEE.

    \bibitem{dong2019learning}
    Dong, X., Thanou, D., Rabbat, M., \& Frossard, P. (2019).
    \textit{Learning graphs from data: A signal representation perspective}.
    IEEE Signal Processing Magazine, 36(3), 44--63.

    \bibitem{elad2010sparse}
    Elad, M. (2010).
    \textit{Sparse and redundant representations: from theory to applications in signal and image processing}.
    Springer Science \& Business Media.

    \bibitem{golub2013matrix}
    Golub, G. H., \& Van Loan, C. F. (2013).
    \textit{Matrix computations}.
    JHU Press.

    \bibitem{uhlig2001constructive}
    Uhlig, F. (2001).
    \textit{Constructive ways for generating (generalized) real orthogonal matrices as products of (generalized) symmetries}.
    Linear Algebra and its Applications, 332, 459--467.

    \bibitem{bao2013fast}
    Bao, C., Cai, J.-F., \& Ji, H. (2013).
    \textit{Fast sparsity-based orthogonal dictionary learning for image restoration}.
    In Proceedings of the IEEE International Conference on Computer Vision (pp. 3384--3391).

    \bibitem{thanou2014learning}
    Thanou, D., Shuman, D. I., \& Frossard, P. (2014).
    \textit{Learning parametric dictionaries for signals on graphs}.
    IEEE Transactions on Signal Processing, 62(15), 3849--3862.

    \bibitem{olshausen1997sparse}
    Olshausen, B. A., \& Field, D. J. (1997).
    \textit{Sparse coding with an overcomplete basis set: A strategy employed by V1?}.
    Vision research, 37(23), 3311--3325.

    \bibitem{engan1999method}
    Engan, K., Aase, S. O., \& Husoy, J. H. (1999).
    \textit{Method of optimal directions for frame design}.
    In 1999 IEEE International Conference on Acoustics, Speech, and Signal Processing. Proceedings. ICASSP99 (Vol. 5, pp. 2443--2446).
    IEEE.

    \bibitem{aharon2006k}
    Aharon, M., Elad, M., \& Bruckstein, A. (2006).
    \textit{K-SVD: An algorithm for designing overcomplete dictionaries for sparse representation}.
    IEEE Transactions on Signal Processing, 54(11), 4311--4322.

    \bibitem{mairal2009online}
    Mairal, J., Bach, F., Ponce, J., \& Sapiro, G. (2009).
    \textit{Online dictionary learning for sparse coding}.
    In Proceedings of the 26th annual international conference on machine learning (pp. 689--696).

    \bibitem{agarwal2014learning}
    Agarwal, A., Anandkumar, A., Jain, P., Netrapalli, P., \& Tandon, R. (2014).
    \textit{Learning sparsely used overcomplete dictionaries}.
    In Conference on Learning Theory (pp. 123--137).
    PMLR.

    \bibitem{arora2014new}
    Arora, S., Ge, R., \& Moitra, A. (2014).
    \textit{New algorithms for learning incoherent and overcomplete dictionaries}.
    In Conference on Learning Theory (pp. 779--806).
    PMLR.

    \bibitem{geng2014local}
    Geng, Q., \& Wright, J. (2014).
    \textit{On the local correctness of $\ell$ 1-minimization for dictionary learning}.
    In 2014 IEEE International Symposium on Information Theory (pp. 3180--3184).
    IEEE.

    \bibitem{schnass2015local}
    Schnass, K. (2015).
    \textit{Local identification of overcomplete dictionaries}.
    J. Mach. Learn. Res., 16(Jun), 1211--1242.

    \bibitem{arora2015simple}
    Arora, S., Ge, R., Ma, T., \& Moitra, A. (2015).
    \textit{Simple, efficient, and neural algorithms for sparse coding}.
    In Conference on learning theory (pp. 113--149).
    PMLR.

    \bibitem{agarwal2016learning}
    Agarwal, A., Anandkumar, A., Jain, P., \& Netrapalli, P. (2016).
    \textit{Learning sparsely used overcomplete dictionaries via alternating minimization}.
    SIAM Journal on Optimization, 26(4), 2775--2799.

    \bibitem{li2017provable}
    Li, Y., \& Liang, Y. (2017).
    \textit{Provable alternating gradient descent for non-negative matrix factorization with strong correlations}.
    In International Conference on Machine Learning (pp. 2062--2070).
    PMLR.

    \bibitem{sun2015complete}
    Sun, J., Qu, Q., \& Wright, J. (2015).
    \textit{Complete dictionary recovery over the sphere}.
    In 2015 International Conference on Sampling Theory and Applications (SampTA) (pp. 407--410).
    IEEE.

    \bibitem{zhai2020complete}
    Zhai, Y., Yang, Z., Liao, Z., Wright, J., \& Ma, Y. (2020).
    \textit{Complete dictionary learning via l4-norm maximization over the orthogonal group}.
    Journal of Machine Learning Research, 21(165), 1--68.

    \bibitem{spielman2012exact}
    Spielman, D. A., Wang, H., \& Wright, J. (2012).
    \textit{Exact recovery of sparsely-used dictionaries}.
    In Conference on Learning Theory (pp. 37--1).
    JMLR Workshop and Conference Proceedings.

    \bibitem{liang2022simple}
    Liang, G., Zhang, G., Fattahi, S., \& Zhang, R. Y. (2022).
    \textit{Simple alternating minimization provably solves complete dictionary learning}.
    arXiv preprint arXiv:2210.12816.

    \bibitem{hoeffding1994probability}
    Hoeffding, W. (1994).
    \textit{Probability inequalities for sums of bounded random variables}.
    The collected works of Wassily Hoeffding, 409--426.
    Springer.

\end{thebibliography}
\end{document}